\documentclass[12pt]{article}
\usepackage{amsmath,amssymb}
\usepackage{float}
\newtheorem{thm}{Theorem}[section]

\newtheorem{assumption}[thm]{Assumption}

\newtheorem{prop}[thm]{Proposition}

\newtheorem{definition}[thm]{Definition}

\newtheorem{example}[thm]{Example}

\newtheorem{remark}[thm]{Remark}

\newtheorem{protocol}[thm]{Protocol}

\newcommand{\M}{\mathcal{M}}

\renewcommand{\M}{\mathcal{M}}

\renewcommand{\O}{\mathcal{O}}

\newcommand{\R}{\mathbb{R}}

\newcommand{\Mat}{\mathrm{Mat}}

\newcommand{\Section}[1]{\section{#1}\setcounter{equation}{0}}

\newcommand{\openbox}{\leavevmode
  \hbox to.77778em{%
    \hfil\vrule
  \vbox to.675em{\hrule width.6em\vfil\hrule}%
  \vrule\hfil}} 
\newcommand{\proofname}{Proof}

\newenvironment{proof}[1][\proofname]{\par\normalfont
  \trivlist\item[\hskip\labelsep\itshape #1:]\ignorespaces
  }{\hspace*{1cm}\hspace*{\fill}\openbox \medskip\endtrivlist}

\title{Remarks on a Tropical Key Exchange System
}

\date{\today}%
\author{Dylan Rudy \qquad Chris Monico \\
  {\small Department of Mathematics and Statistics\vspace{-2mm}}\\
  {\small Texas Tech University\vspace{-2mm}}\\
  {\small {\em e-mail:\/} c.monico@ttu.edu }
  }

\begin{document}\maketitle
\thispagestyle{empty}
\begin{abstract}
  We consider a key-exchange protocol based on matrices over a tropical
  semiring which was recently proposed in \cite{grig19}. We show that
  a particular private parameter of that protocol can be recovered with a simple
  binary search, rendering it insecure.
\end{abstract}

%%%%%%%%%%%%%%%%%%%%%%%%%%%%%%%%%%%%%%%%%%%%%%%%%%%%%%%%%%%%%%
\Section{Introduction}        \label{sec:intro}

  Let $S$ be any nonempty subset of $\R$ which is closed under addition.
  Define two operations $\oplus$ and $\otimes$ on $S$ by
  \begin{eqnarray*}
    a\oplus b &=& \min\{a, b\}, \\
    a\otimes b &=& a+b.
  \end{eqnarray*}
  Both operations are associative and commutative and $\otimes$ distributes over $\oplus$,
  and hence $S$ is a commutative semiring, called a {\em tropical semiring}. 
  The set $\M = \Mat_{k\times k}(S)$ of $k\times k$ matrices over $S$ is therefore a
  semiring with the induced operations
  \begin{eqnarray*}
    (a_{ij}) \oplus (b_{ij}) &=& (a_{ij}\oplus b_{ij}),\\
    (a_{ij}) \otimes (b_{ij}) &=& (c_{ij}), \hspace{12pt} \mbox{ where }
        c_{ij} = (a_{i1}\otimes b_{1j}) \oplus (a_{i2}\otimes b_{2j}) \oplus \dots \oplus (a_{ik}\otimes b_{kj}).
  \end{eqnarray*}

  In \cite{grig14}, the authors proposed two key exchange protocols based on the structure $\M$.
  Shortly after, an effective attack was given on one of those protocols in \cite{koto18}.
  Subsequently, a new key exchange protocol was proposed in \cite{grig19} (in fact, two new
  protocols, but they are very closely related to each other). It is this protocol that we
  consider in this paper.

  In \cite{grig19}, the authors 
  give two semigroup operations on $\M\times \M$
  each arising as a semidirect product induced by a specified action of these matrices on
  themselves.
  The two semigroup operations are given by
  \begin{eqnarray} \label{eq:op1}
    (M,G)\circ (S,H) &=& \Big(M\oplus S \oplus H\oplus (M\otimes H),\hspace{6pt} G\oplus H\oplus (G\otimes H)\Big),\\
    \label{eq:op2}
    (M,G) * (S,H)    &=& \Big( (H\otimes M^T)\oplus (M^T\otimes H)\oplus S, \hspace{6pt} G\otimes H \Big).
  \end{eqnarray}
  Note that for each of these operations, the first component of the product does not depend on $G$.
  This fact plays a key role in the two key exchange protocols they then propose (one corresponding
  to each operation):

% CJM: I added the `... and they agree on a positive integer $K$.',
% Because it was previously undefined.
  \begin{enumerate}
    \item Alice and Bob agree on public matrices $M,H\in \M$ whose entries
          are integers in the range $[-N,N]$, and they agree on a positive integer $K$.
          Alice selects a private positive integer $m<2^{K}$ and Bob selects a private positive integer $n<2^{K}$.
    \item Alice computes $(M,H)^m = (A, P_A)$ and sends $A$ to Bob.
    \item Bob computes $(M,H)^n = (B, P_B)$ and sends $B$ to Alice.
    \item Alice determines the first component of $(M,H)^{m+n} = (M,H)^n(M,H)^m = (B, P_B)(A, P_A)$ 
          from her knowledge of $A, P_A,$ and $B$ (knowledge of $P_B$ is not necessary
          for either of the operations \eqref{eq:op1} or \eqref{eq:op2}.
    \item Bob similarly determines the first component of $(M,H)^{m+n} = (M,H)^m(M,H)^n = (A, P_A)(B, P_B)$ 
          from his knowledge of $B, P_B,$ and $A$.
  \end{enumerate}

  In the next section, we show that an eavesdropper can find a positive integer $m'$
  for which the first component of $(M,H)^{m'}$ is $A$; she can then use this $m'$
  to compute the shared secret key in essentially the same way as Alice.
  Furthermore, such an $m'$ can be found using $\O(K^2)$ operations
  \eqref{eq:op1} or \eqref{eq:op2}.

%%%%%%%%%%%%%%%%%%%%%%%%%%%%%%%%%%%%%%%%%%%%%%%%%%%%%%%%%%%%%%
\Section{The attack}        \label{sec:attack}

  Since addition of matrices in $\M$ is idempotent, i.e., $G\oplus G = G$,
  we have a partial order on $\M$ defined by
  \[
    X \le Y \hspace{12pt}\mbox{ if } X\oplus Y = X.
  \]
  Clearly we have that $X\le Y$ iff $x_{ij} \le y_{ij}$ for all $i,j\in\{1,2,\dots, k\}$.
  Furthermore, this partial order respects both operations on $\M$; if
  $X\le Y$ and $Z\in\M$, then $X\oplus Z \le Y\oplus Z$ and $X\otimes Z \le Y\otimes Z$.

  \begin{prop}
    Consider the semigroup $\M\times \M$ equipped with either of the two operations
    defined by \eqref{eq:op1} and \eqref{eq:op2}. 
    Let $(M,H)\in \M\times\M$,
    and for each positive integer $\ell$ let $(M_\ell, H_\ell) = (M,H)^\ell$.
    Then the sequence $\{M_\ell\}$ is monotonically decreasing:
    $M_1\ge M_2 \ge M_3 \ge \dots$.
  \end{prop}
  \begin{proof}
    Let $\ell \ge 2$.
    For the operation $\circ$ we have
    \begin{eqnarray*}
      (M_\ell, H_\ell) &=& (M_{\ell-1}, H_{\ell-1})\circ (M,H) \\ 
                       &=& \Big( M_{\ell-1}\oplus M\oplus H \oplus (M_{\ell-1}\otimes H), H_{\ell-1}\oplus H\oplus (H_{\ell-1}\otimes H)\Big),
    \end{eqnarray*}
    so that $M_\ell = M_{\ell-1}\oplus M\oplus H \oplus (M_{\ell-1}\otimes H)$.
    In particular, $M_{\ell} \oplus M_{\ell-1} = M_\ell$, and hence $M_\ell \le M_{\ell-1}$.

    Similarly, for the operation $*$ we have that
    \begin{eqnarray*}
      (M_\ell, H_\ell) &=& (M,H)* (M_{\ell-1}, H_{\ell-1}) \\ 
                       &=& \Big( (H_{\ell-1}\otimes M^T)\oplus (M^T\otimes H_{\ell-1})\oplus M_{\ell-1}, H\otimes H_{\ell-1}\Big),
    \end{eqnarray*}
    and hence $M_\ell = (H_{\ell-1}\otimes M^T)\oplus (M^T\otimes H_{\ell-1})\oplus M_{\ell-1}$.
    Again, $M_{\ell}\oplus M_{\ell-1} = M_\ell$, so that $M_\ell \le M_{\ell-1}$.
  \end{proof}

  The problem alluded to at the end of the introduction is now easily solved
  with a binary search. Let $M,H\in\M$ and $(M,H)^\ell = (M_\ell, H_\ell)$.
  Suppose $A\in\M$ satisfies $A=M_m$ for some positive integer $m<2^K$.
  First, obtain an upper bound on $m$ by computing successive squares
  \[
    M_1, M_2, M_4, M_8,\dots
  \]
  until finding a positive integer $t$ for which $A\le M_{2^t}$.
  Since it is then known that $1\le m \le 2^t$, a simple binary
  search will find an integer $m'$ for which $M_{m'} = A$.
  The sequence $M_1,M_2,\dots$ is generally strictly decreasing, in which
  case $m'=m$. However, even if $m'\ne m$, finding such an integer $m'$
  is enough for the eavesdropper to recover the shared secret key.
  Let $\pi_1:\M\times\M \longrightarrow \M$ be the map $\pi_1(C,D)=C$.
  Suppose $(M,H)^n = (B, P_B)$, $(M,H)^m = (A,P_A)$ and $(M,H)^{m'} = (A, P_E)$.
  Then for each of the operations \eqref{eq:op1} and \eqref{eq:op2},
  the shared secret key satisfies
  \[
    \pi_1( (M,H)^{m+n}) = \pi_1( (M,H)^{m'+n}).
  \]
  This is clear, since this shared secret key can be expressed in terms of $A,B$, and $P_B$ only,
  but it may also be explicitly verified. For example, with the operation \eqref{eq:op1},
  \begin{eqnarray*}
    \pi_1( (M,H)^{m+n}) &=& \pi_1( (A,P_A)\circ (B, P_B)) \\ 
              &=& A\oplus B\oplus P_B \oplus (A\otimes P_B) \\
              &=& \pi_1( (A,P_E)\circ (B, P_B) )\\
              &=& \pi_1( (M,H)^{m'+n}).
  \end{eqnarray*}
  In particular, the eavesdropper may recover the shared secret key via
  \begin{eqnarray*}
    \pi_1( (M,H)^{m+n}) &=& \pi_1( (M,H)^n \circ (M,H)^{m'}) \\ 
          &=& \pi_1( (B, P_B)\circ (A,P_E)) \\
          &=& B\oplus A\oplus P_E \oplus (B\otimes P_E).
  \end{eqnarray*}

  Finding $t$ as described above requires at most $K$ semigroup operations in $\M\times\M$.
  The binary search, done in the most obvious way, would compute $K$
  powers of $(M,H)$, each of which requires no more than $2K$ semigroup
  operations in $\M\times\M$, for a total complexity of at most $2K^2+K$ operations
  in $\M\times\M$. This can be reduced to $K^2+K$ by storing the successive
  squares $(M_1,H_1), (M_2,H_2), (M_4, H_4),\ldots$ and using them to compute
  each power of $(M,H)$ during the binary search phase.

  Addition of $k\times k$ matrices can be accomplished with $\O(k^2)$
  integer max operations, and
  multiplication accomplished using $\O(k^3)$ integer addition and max operations.
  Therefore this attack requires $\O(K^2k^3)$ integer operations.
  We argue below that the typical entry of $A$ has about
  $K$ bits. In that case, each integer addition and max operation requires
  no more than $K$ bit operations, for a total of $\O(K^3k^3)$ bit
  operations.
  If we let $\alpha$ denote the number of bits required to represent $A$ (i.e.,
  the key size)
  it follows that $\alpha\approx Kk^2$, and this attack requires $\O(\alpha^3)$
  bit operations, a polynomial-time function of the input size. If $K$ is fixed,
  as in our experiments, then it requires $\O(\alpha^{1.5})$ bit operations.

  We coded this method in C, and performed some experiments on
  a single core of an i7 CPU at 3.10GHz. Using $\M = \Mat_{k\times k}(S)$
  for various values of $k$, and the parameters $N=1000$, $K=200$ suggested
  in \cite{grig19}, we performed 40 experiments for each value of $k$. In each
  experiment, we generated random matrices $M,H$ and chose random positive
  integers $m,n<2^K$ and measured the time to recover an $m'$ as described
  above. The results of these experiments are summarized in Table \ref{tab:exp1}.
  For reference, we also report the average number of bits $\alpha$ in the matrix $A$ 
  that would be shared by Alice, and the values $t/k^3$ and $t/\alpha^{1.5}$ for
  comparison with the asymptotic runtime estimates given above.

  \begin{table}[H]
  \begin{center}\begin{tabular}{|rrrrr|}
    \hline
      $k$ & $\alpha$ & $t$ & $t/k^3$ & $t/\alpha^{1.5}$  \\
     \hline
      5 & 5222 & 0.12 &     0.00096  & 3.2e-7 \\
      10 & 20885 & 0.66 &   0.00066  & 2.2e-7  \\
      15 & 47025 & 2.43 &   0.00072  & 2.4e-7 \\
      20 & 83710 & 4.76 &   0.00060  & 2.0e-7 \\
      25 & 130594 & 10.53 & 0.00067  & 2.2e-7 \\
      30 & 188145 & 17.75 & 0.00066  & 2.2e-7 \\ 
      35 & 256484 & 24.05 & 0.00056  & 1.9e-7 \\
      40 & 334040 & 40.92 & 0.00064  & 2.1e-7 \\
      45 & 422111 & 45.80 & 0.00050  & 1.7e-7 \\
      50 & 523312 & 78.33 & 0.00063  & 2.1e-7 \\
      55 & 631091 & 98.19 & 0.00059  & 2.0e-7\\
      60 & 752490 & 122.57& 0.00057  & 1.9e-7 \\
    \hline 
  \end{tabular}\end{center}
  \caption{Average number of bits $\alpha$ to represent $A$ (Alice's matrix, from Section 1), 
           and average time $t$ (in seconds) to recover $m'$ 
           for various sized ($k\times k$) matrices, with $N=1000$ and $K=200$.}
  \label{tab:exp1}
  \end{table}

  We would like to make one final remark about the key sizes in this system.
  With the notation as above and the operation \eqref{eq:op1}, for example,
  we have
  \[
    M_{\ell+1} = M_\ell\oplus M\oplus H \oplus (M_l\otimes H).
  \]
  Since $M_2\le M$ and $M_2\le H$ and $M_{\ell+1} \le M_2$ for all $\ell\ge 2$, it follows that
  \[
    M_{\ell+1} = M_\ell \oplus (M_\ell \otimes H), \hspace{12pt} \mbox{ for } \ell\ge 2.
  \]
  This means that, on average, the entries of $M_{\ell+1}$ decrease from those of $M_\ell$
  by an approximately constant amount, proportional to the size of the entries of $H$.
  With Alice's $m\approx 2^K$,
  this means that the entries of $A$
  are on the order of $-c\times 2^K$, or about $K$ bits each.
  With the parameter sizes $K=200$, $k=30$, $N\approx 1000$ suggested in
  \cite{grig19}, one would have $M$ and $H$ consisting of about 9000 bits
  each and $A$ with about $30\times 30\times 200 = 180,000$ bits.

%%%%%%%%%%%%%%%%%%%%%%%%%%%%%%%%%%%%%%%%%%%%%%%%%%%%%%%%%%%%%%
\Section{Conclusion}        \label{sec:conclusion}

  The attack presented here 
exploits the fact that the sequence $\{(M,H)^\ell\}$ is linearly ordered.
It is quite effective and practical against the protocols described in \cite{grig19}.
For those protocols, Alice and Bob must do approximately $\O(K)$
operations in the semigroup $\M\times\M$, and this attack requires
about  $\O(K^2)$ operations in that same semigroup, so an increase
of parameter sizes does not help.

  We thank the referees for their thoughtful reading of this manuscript
and their feedback.

%%%%%%%%%%%%%%%%%%%%%%%%%%%%%%%%%%%%%%%%%%%%%%%%%%%%%%%%%%%%%%%%%%%%
%\bibliography{trop_crypt}
\bibliographystyle{plain}

\end{document}